\documentclass[11pt]{article}

\RequirePackage[l2tabu,orthodox]{nag}

\RequirePackage[utf8]{inputenc}
\RequirePackage[T1]{fontenc}

\RequirePackage{lmodern}

\usepackage{amsmath}
\usepackage{amsthm}
\usepackage{amsfonts}

\usepackage{multirow}
\usepackage{bigstrut}
\usepackage{xspace}
\usepackage{graphicx}
\usepackage{comment}
\usepackage{cite}
\usepackage{color}

\usepackage[boxed,linesnumbered,vlined]{algorithm2e}

\newtheorem{theorem}{Theorem}

\newcommand{\remove}[1]{}
\newcommand{\ignore}[1]{}

\newcommand{\concept}[1]{\textbf{#1}}

\SetKwInput{LocalData}{local data}
\SetKwInput{SharedData}{shared data}
\SetKwInput{Constants}{constants}
\SetKw{Break}{break}
\SetKwData{MyID}{pid}
\SetKw{True}{true}
\SetKw{False}{false}
\SetKw{KwAnd}{and}
\SetKw{Not}{not}
\SetKwIF{If}{ElseIf}{Else}{if}{then}{else if}{else}{end if}
\SetKwFor{For}{for}{do}{end for}
\SetKwFor{While}{while}{do}{end while}

\SetKwFor{Procedure}{procedure}{}{end procedure}

\SetKwIF{WithProbability}{ElseIf}{Else}{with probability}{do}{else if}{else}{end}

\newcommand{\newData}[2]{\newcommand{#1}{\DataSty{#2}\xspace}}
\newcommand{\newFunc}[2]{\newcommand{#1}{\FuncSty{#2}\xspace}}

\newData{\MyId}{myId}

\newFunc{\AdoptCommit}{AdoptCommit}
\newData{\Adopt}{adopt}
\newData{\Commit}{commit}

\newData{\maxRound}{maxRound}
\newData{\Win}{win}
\newFunc{\TAS}{TAS}

\title{A one-bit swap object using test-and-sets\\ and a max register}

\author{
James Aspnes\thanks{Yale University, Department of Computer Science.  Supported in part by NSF grant CCF-0916389.
\texttt{aspnes@cs.yale.edu}.}
}

\newFunc{\Swap}{swap}

\begin{document}

\maketitle

\begin{abstract}
We describe a linearizable, wait-free implementation of a one-bit
\Swap object from a single max register and an unbounded array of
test-and-set bits.  Each \Swap operation takes at most three
steps.  Using standard randomized constructions, the max register and
test-and-set bits can be replaced by read-write registers, at the
price of raising the cost of a \Swap operation
to an expected $O(\max(\log n, \min(\log t, n)))$ steps, where $t$
is the number of times the \Swap object has previously
changed its value and $n$ is the number of processes.
\end{abstract}

\section{Introduction}
\label{section-introduction}

A \Swap object supports a single read-modify-write operation \Swap
that returns the old contents of the object while setting a new value.
The simplest variant of a \Swap object is one that stores only a
single bit.  This variant is equivalent to a test-and-set object that
has been extended with a test-and-reset operation, where each
operation returns the old value of the object and writes a new value
($1$ for test-and-set and $0$ for test-and-reset), all as an atomic
operation.

General implementations of \Swap objects can be very expensive, even
given test-and-set bits.  The best known general \Swap object
implementation is that
of Afek, Weisberger, and Weisman~\cite{AfekWW1992}, which may require
as many as $\Theta(n \log n)$ steps to carry out a single \Swap
operation even in 
the one-shot case.  Whether this cost can be reduced is an interesting
open question.

We do not answer this question, but instead observe that the cost can
be greatly reduced if the size of the \Swap object is restricted to a
single bit.  We give a simple implementation of a \Swap object from a
single max register~\cite{AspnesAC2012} that indexes an unbounded array of test-and-set
bits.  The key observation is that \Swap operations on a one-bit
register can be linearized by first by separating out groups of
\Swap operations that all have the same input $0$ or $1$
(using the max register), and then
choosing a single operation from each group to linearize first (using
a test-and-set).  Because the \Swap object is limited to one bit,
knowing whether an operation is linearized first within its group is
enough to determine its return value: it will be equal to the common
input of the group if it is \emph{not} linearized first and equal to
the other input if it is.  No further ordering of operations within a
group is needed.

It is known~\cite{AspnesAC2012} that unbounded max registers can be
implemented directly from read-write registers, at a cost of
$O(\min(\log v, n))$ steps for any operation that leaves a max register
with value $v$.  Test-and-set bits can also be implemented from
read-write registers if randomization is
permitted; the costs of the best current implementations are 
an expected
$O(\log n)$ register operations for each test-and-set operation 
assuming an adaptive
adversary
that can react to what the implementation does~\cite{AfekGTV1992} and
$O(\log^* n)$ expected operations assuming an
oblivious adversary that cannot~\cite{GiakkoupisW2012}.
Applying these construction to our
algorithm gives a cost of either 
$O(\max(\log n, \min(\log t, n)))$ 
or
$O(\max(\log^* n, \min(\log t, n)))$ 
register operations on average
for each \Swap operation, where $t$ is the number of times the \Swap
object switches between its two values in the linearized schedule.
For typical values of $t$, we would expect the $O(\log t)$ term to
dominate.

\section{Model}

We assume a standard asynchronous shared-memory model, with
concurrency modeled by interleaving under the control of an
\concept{adversary scheduler}.  We are interested in
implementations of objects that are
\concept{wait-free} (every process
finishes in a finite number of steps in any execution) and
\concept{linearizable}~\cite{HerlihyW1990} (there exists a sequential
execution of the object that is consistent with the observed execution
order).  

Our base objects consist of a max register and an array of
test-and-set bits.  A \concept{max register}~\cite{AspnesAC2012}
supports write 
and read operations, where a read operation returns the largest
value previously written.  A \concept{test-and-set} bit supports a
single operation $\TAS$, which sets the bit to $1$ and returns the
previous value.  Unless otherwise specified, 
we assume that both the max register and the
test-and-set bits are initialized to $0$.  As discussed previously,
we can 
also use standard techniques to replace these base objects with
ordinary registers.

\section{Implementation}

Pseudocode for the \Swap operation is given in Algorithm~\ref{alg-swap}.
The implementation uses a single max register $\maxRound$, and an
unbounded array of test-and-set bits $t[0\dots]$.  To initialize the
\Swap object to $b$, set $\maxRound$ to $b$ and initialize $t[b]$ to
$1$ (as if a \TAS operation had already successfully been performed on
it); this is equivalent to running $\Swap(b)$ with $\maxRound$ and all
test-and-set objects initialized to $0$
and discarding the result.

\begin{algorithm}
    \Procedure{$\Swap(v)$}{
        $r \leftarrow \maxRound$ \\
        \If{$r \not\equiv v \pmod 2$}{
            \label{line-parity-check}
            $r \leftarrow r+1$ \\
            $\maxRound \leftarrow r$
        }
        \eIf{$\TAS(t[r]) = 0$}{
            \label{line-tas}
            \Return $\neg v$
        }{
            \Return $v$
        }
    }
    \caption{Pseudocode for a \Swap operation}
    \label{alg-swap}
\end{algorithm}

The step complexity of this implementation is
$O(1)$.  Indeed, each execution of $\Swap$ requires either two or three
operations on the base objects depending on the outcome of the test in
Line~\ref{line-parity-check}.  

Both max registers and test-and-set bits can be implemented from
registers.
If the max register $r$ is implemented
from registers using the technique of~\cite{AspnesAC2012}, the cost 
becomes
$O(\log v, n)$, where $v$ is the value in the max register.  It is easy to
see that $v$ is bounded by the number of $\Swap$ operations, since
each $\Swap$ operation increments it at most once.  
Test-and-set bits can also be implemented directly from registers
using randomization.  Using the best currently-known implementations,
the cost is an expected $O(\log n)$ steps per
test-and-set operation~\cite{AfekGTV1992} assuming an adaptive
adversary and 
$O(\log^* n)$~\cite{GiakkoupisW2012} assuming an oblivious adversary.  In either case the cost of the test-and-set will be
dominated by the cost of the max register after a linear number of
$\Swap$ operations in the worst case.

\section{Linearizability}
\label{section-linearizability}

To show linearizability, we construct an explicit linearization order
based on the final value of $r$ for each \Swap operation, with
processes sharing the same value ordered further by the linearization order of
the test-and-set bit $t[r]$.

\begin{theorem}
Algorithm~\ref{alg-swap} is a linearizable implementation of a \Swap
object.
\end{theorem}
\begin{proof}
Fix an execution of the protocol.

For each \Swap operation $\sigma$, define $r(\sigma)$ to be the value
of the internal variable $r$ at the time of the call to
$\TAS(t[r])$ in Line~\ref{line-tas} of the execution of $\sigma$.
Note that $r(\sigma) \bmod 2$ is always equal to the input value
$v_\sigma$
of $\sigma$.
Let $S_i$ be the
set of all \Swap operations $\sigma$ for which $r(\sigma) = i$.  We
will construct a linearized execution by ordering the sets $S_i$
by increasing $i$, and ordering operations within each $S_i$ based on
the linearization order for $t[i]$.

To show that this is in fact a linearization, we must show both that
it respects the observable order of operations and that the resulting
execution corresponds to a sequential execution of a \Swap object.

For the first part, suppose that some operation $\sigma_1$ finishes
before another operation $\sigma_2$ starts.  
First let us show that $r(\sigma_1) \le r(\sigma_2)$.
The value $r(\sigma_1)$ is either read from $\maxRound$ or written to
it before $\sigma_1$ finishes; the subsequent read of $\maxRound$ by
$\sigma_2$ thus returns a value $r' \ge r(\sigma_1)$, and 
$r(\sigma_2)$ is either $r'$ or $r'+1$, which in either case is
greater than or equal to $r(\sigma_1)$.
If $r(\sigma_1) < r(\sigma_2)$, then the two operations are in
distinct sets $S_{r(\sigma_1)}$ and $S_{r(\sigma_2)}$, and $\sigma_1$
is linearized first.
If instead $r(\sigma_1) = r(\sigma_2)$, then both are in the same set
$S_i$.
Now because 
$\sigma_1$ accesses
$t[i]$ before $\sigma_2$, it again holds that $\sigma_1$ is linearized
first.  

For the second part, we start by showing that there are no gaps in the
sequence of sets $S_i$.  
Specifically, we observe that if $S_i$ is nonempty
for $i > b+1$, where $b$ is the initial value of $\maxRound$,
then so is $S_{i-1}$.  The reason is that if $S_i$ is nonempty, then
either some operation reads $i$ from $\maxRound$ or writes $i$ to
$\maxRound$.  In either case, because $i$ is not the initial value of
$\maxRound$, there is a first operation $\sigma$ that
writes $i$ to $\maxRound$.  This operation must previously have read
$i-1$ from $\maxRound$.  Since $i > b+1$, $i-1 > b$, and 
so $i-1$ can only appear in $\maxRound$ if some other operation
$\sigma'$ writes it.  But then $\sigma' \in S_{i-1}$ and
$S_{i-1}$ is nonempty as claimed.

Now consider some specific operation $\sigma$ and let $i = r(\sigma)$.
Recall that $i \bmod 2 = v_\sigma$, where $v_\sigma$ is the input to $\sigma$.
There are two cases, depending on the value returned by $\TAS(t[r])$
in $\sigma$:
\begin{itemize}
\item If this value is $0$, then we have that (a) $\sigma$
is linearized first among all operation in $S_i$, and (b)
$\sigma$ returns $\neg v_\sigma = (i-1) \bmod 2$.  If $S_{i-1}$ is nonempty,
then there exists a $\Swap(\neg v_\sigma)$ operation in $S_{i-1}$ that
linearizes immediately before $\sigma$, and thus it is correct for
$\sigma$ to return $\neg v_\sigma$.  If $S_{i-1}$ is empty, then $i-1 \le b$.
It cannot be the case that $i = b$, because $t[b]$ is initialized to
$1$, contradicting the assumption that $\TAS(t[i])$ returns $0$.  
Nor can we have $i < b$.  It follows that $i-1 = b$, and $\sigma$
correctly returns the initial value $b$.
\item If this value is $1$, then either (a) $\sigma$ is not linearized as
    the first operation in $S_i$, or (b) $\sigma$ is linearized as the
    first operation in $S_i$ and $i=b$.  In the first case, $\sigma$
    returns the input to the previous operation in $S_i$; in the
    second, it returns the initial value $b$.  In both cases this
    return value is correct.
\end{itemize}
\end{proof}

\section{Conclusion}

We've shown that it is possible to build a very efficient \Swap object
from test-and-set bits and max registers, if the \Swap object is
limited to two values.  The key idea is that we can alternate
sequences of $\Swap(0)$ and $\Swap(1)$ operations
so that the operations within each
sequence can be linearized with a single test-and-set bit.  Because
there are only two possible values, the return value of each \Swap
operation can be computed directly from the result of the test-and-set
operation: either it is linearized after another \Swap with the same
input, or it is linearized after another \Swap with a different input.
Unfortunately, there does not seem to be any direct way to expand this
trick to handle more than two inputs.

From the work of Afek, Weisberger, and Weisman~\cite{AfekWW1992}, we
know that a general \Swap object can be implemented directly from
test-and-set bits and read-write registers, but the cost per swap
operation is superlinear in the number of processes.  This leaves a huge complexity
gap between the two-valued case and the general case.  A natural next
step might be to look at less restricted cases such as three-valued \Swap.
This object is general enough to break the specific technique used
here for two-valued swap, but may still allow for a highly efficient
implementation.

\section{Acknowledgments}

The question of how to build small swap objects was inspired by
discussions of a related problem with Dan Alistarh.  I would like to
thank Dan Alistarh, Faith Ellen, and Keren Censor-Hillel for comments
on the algorithm and discussions of possible extensions.

\bibliographystyle{alpha}
\bibliography{paper}

\end{document}